%% file: AllReduce-Scheduling-INFOCOM_21.tex
\documentclass[10pt,conference,letterpaper]{IEEEtran}

\usepackage{cite}
\usepackage{latexsym}
\usepackage{amsmath}
\usepackage{amsthm}
\usepackage{amssymb}
\usepackage{mathrsfs}
\usepackage[dvips,dvipdf]{graphicx}
\usepackage{algorithmicx}
\usepackage{multirow}
\usepackage{bbm}
\usepackage{arydshln}
\usepackage{url}
\usepackage{pict2e,picture} 
\usepackage{tikz}   
\usepackage{pifont} 
\usepackage{enumitem}
\usepackage{mathtools}
\usepackage[noend]{algpseudocode}
\usepackage{comment}
\usepackage{xcolor}
\usepackage[ruled,vlined,linesnumbered,noend]{algorithm2e}
\usepackage{subcaption}
\usepackage{dsfont}

\usepackage{caption}
\allowdisplaybreaks[1]

\input{Supporting-Preambles/symbols-commands}

\IEEEoverridecommandlockouts

\begin{document} 

\newcommand{\kliu}[1]{{\color{blue}{\bf Kevin:} #1}}

\title{GADGET: Online Resource Optimization for Scheduling Ring-All-Reduce Learning Jobs}

\author{Menglu Yu$^{1}$ \mbox{\hspace{0.4cm}} Ye Tian$^{1}$ \mbox{\hspace{0.4cm}} Bo Ji$^{2}$ \mbox{\hspace{0.4cm}} Chuan Wu$^{3}$ \mbox{\hspace{0.4cm}} Hridesh Rajan$^{1}$  \mbox{\hspace{0.4cm}} Jia Liu$^{4,1}$
\\ $^{1}$Department of Computer Science, Iowa State University
\\ $^{2}$Department of Computer Science, Virginia Tech
\\ $^{3}$Department of Computer Science, The University of Hong Kong
\\ $^{4}$Department of Electrical and Computer Engineering, The Ohio State University
\thanks{
This work has been supported in part by NSF grants CAREER CNS-2110259, CNS-2112471, CNS-2102233, CCF-2110252, ECCS-2140277, CNS-2112694, CCF 1934884, CNS 2120448, a Google Faculty Research Award, and Hong Kong RGC grants HKU 17204619, 17208920, 17207621.
}
}


\maketitle

\input{Abstract/Abstract}
\input{Sec1-Intro/Sec1-Intro}
\input{Sec2-Related/Sec2-Related}

\input{Sec3-Preliminaries/Sec3-Preliminaries}

\input{Sec4-Model/Sec4-Model}

\input{Sec5-Algorithm/Sec5-Algorithm}

\input{Sec6-Numerical/Sec6-Numerical}

\input{Sec7-Conclusion/Sec7-Conclusion}

\clearpage
\bibliography{IEEEabrv,./BIB/JobScheduling, ./BIB/DMLF, BIB/ApproxAlg, ./BIB/CSMA}
%

\end{document}

%% file: Supporting-Preambles/symbols-commands.tex
\newcommand{\g}{\mathbf{g}}

\newcommand{\w}{\mathbf{w}}

\newcommand{\x}{\mathbf{x}}

\newcommand{\y}{\mathbf{y}}

\newtheorem{thm}{Theorem}

\newtheorem{lem}[thm]{Lemma}

\newtheorem{defn}[thm]{Definition}

\algtext*{EndWhile}
\algtext*{EndIf}
\algtext*{EndFor}
\algtext*{EndFunction}
\algtext*{EndProcedure}


%% file: Abstract/Abstract.tex

\begin{abstract}
Fueled by advances in distributed deep learning (DDL), recent years have witnessed a rapidly growing demand for resource-intensive distributed/parallel computing to process DDL computing jobs.
To resolve network communication bottleneck and load balancing issues in distributed computing, the so-called ``ring-all-reduce'' decentralized architecture has been increasingly adopted to remove the need for dedicated parameter servers.
To date, however, there remains a lack of theoretical understanding on how to design resource optimization algorithms for efficiently scheduling ring-all-reduce DDL jobs in computing clusters.
This motivates us to fill this gap by proposing a series of new resource scheduling designs for ring-all-reduce DDL jobs.
Our contributions in this paper are three-fold:
i) We propose a new resource scheduling analytical model for ring-all-reduce deep learning, which covers a wide range of objectives in DDL performance optimization (e.g., excessive training avoidance, energy efficiency, fairness);
ii) Based on the proposed performance analytical model, we develop an efficient resource scheduling algorithm called GADGET (\underline{g}reedy ring-\underline{a}ll-reduce \underline{d}istributed \underline{g}raph \underline{e}mbedding \underline{t}echnique), which enjoys a provable strong performance guarantee;
iii) We conduct extensive trace-driven experiments to demonstrate the effectiveness of the GADGET approach and its superiority over the state of the art.
\end{abstract}

%% file: Sec1-Intro/Sec1-Intro.tex

\section{Introduction} \label{sec:intro}

In recent years, the rise of complex deep learning applications has led to a rapidly growing demand for resource-intensive (e.g., GPUs, memory, energy) distributed/parallel computing to process deep learning training tasks.
Traditionally, most distributed deep learning (DDL) frameworks are based on the parameter server (PS)-worker architecture, which consists of a set of PS(s) and workers. 
Despite its simplicity, the PS-worker architecture suffers from two scalability limitations:
i) The topology of the PS-worker architecture creates a {\em communication bottleneck} at each PS as the number of workers increases; and
ii) The centralized PSs are vulnerable to the {\em single-point-of-failure risk}.
To overcome these scalability weaknesses, the more sophisticated {\em ``ring-all-reduce''} (RAR) parallel computing architecture has become increasingly popular for DDL training and has been supported by many mainstream DDL frameworks (e.g., Tensorflow~\cite{Abadi16:TensorFlow}, Pytorch~\cite{Paszke19:Pytorch}).
Specifically, by forming a ring between the workers to jointly perform parameter sharing and reduction, the RAR  architecture removes the need for dedicated PS(s), hence alleviating the single point of failure. 

However, with the increasing adoption of the RAR architecture for DDL training, an important question naturally emerges: {\em How could we design 
resource optimization algorithms to efficiently schedule RAR-based DDL training jobs over networked computing clusters?}
Answering this question is critical because:
i) Multi-core high-throughput GPU hardware for cloud-based machine learning services is expensive, which requires efficient GPU utilization.
For example, an Amazon EC2 eight-core GPU instance with NVLink connection costs more than \$31 per hour\cite{P3_COST:2021};
ii) Due to the resource competition among multiple RAR-based DDL training jobs in the cluster,
if scheduling is not done strategically, DDL training jobs could suffer from large latency;
and iii) RAR-based DDL jobs are often trained by variants of the {\em iterative} stochastic gradient descent (SGD) method to optimize thousands or even millions of parameters. 
Their completion times are defined by the convergence processes of these SGD-based methods, which often exhibit the {\em ``diminishing return effect''} in terms of the gains in training accuracy as the number of iterations increases.
For instance, it has been shown in \cite{Jeon2018:multi-tenant} that approximately 75\% of DDL jobs reach within 0.1\% of the lowest training loss using only 40\% of the epochs.
Allocating too many computing resources to one job not only induces high GPU and energy costs as well as unfairness to other DDL jobs, but also leads to unnecessary training time with little learning accuracy performance gain in return.

However, optimizing scheduling and resource allocation for RAR-based DDL is highly non-trivial due to several technical challenges.
{\em First}, the hop-by-hop dependence in the ring structure renders the placement of workers highly sensitive to intra- and inter-server communications, decided by the underlying computing network topology.
Also, as will be shown later, the ring topology violates the loop-less assumption of many existing algorithms for virtual computing resource allocation problems \cite{Rost18:VNE, Avin18:CoRR, Rost19:Sigcomm}, which necessitates new algorithmic design techniques.
{\em Second}, the 
resource allocation for each RAR-based DDL job is subject to packing-type constraints (due to resource limits), which implies NP-Hardness.
{\em Lastly}, the scheduler is not aware of DDL job arrivals beforehand, which calls for {\em online} optimization algorithm design.
Perhaps due to these challenges, to date, results on scheduling and resource allocation for RAR-based DDL training remain limited.
This motivates us to fill this gap by proposing an efficient resource scheduling algorithm called GADGET (\underline{g}reedy ring-\underline{a}ll-reduce \underline{d}istributed \underline{g}raph \underline{e}mbedding \underline{t}echnique), which addresses the aforementioned technical challenges with strong theoretical performance guarantees.
The main results in this paper and their significance are summarized as follows:

\begin{list}{\labelitemi}{\leftmargin=1em \itemindent=0em \itemsep=.2em}

\item By extracting the key architectural features of RAR-based DDL training jobs, we develop a new analytical model for scheduling RAR-based DDL jobs over networked environments.
Based on this model, we formulate a general online performance optimization framework for RAR-based DDL training.
We note that, due to the heterogeneous internal (i.e., between containers in the same physical server) and external (i.e., between physical machines within or beyond the same rack) communications, the resource scheduling and allocation problem for RAR-based DDL training is far more challenging than those for the PS-worker architecture.

\item To address the new challenges arising from RAR-based DDL training, we develop an efficient GADGET algorithm that can provide provably strong performance guarantee.
Our GADGET approach is based on a ``divide-and-conquer'' approach.
Specifically, we first show that the formulated RAR-based training scheduling problem over the temporal domain possesses a submodular structure due to its partition matroid nature, and hence can be decomposed in the temporal domain and solved by a greedy approach with worst-case approximation ratio guarantee.

\item Next, we focus on the decomposed subproblem in each time-slot, which remains NP-Hard due to the packing-type constraints.
We note that the ring-structure of RAR-based DDL jobs renders existing methods ineffective.
To address this challenge, we propose a {\em generalized} virtual graph embedding (G-VNE) technique that guarantees a $\frac{1}{3\Gamma}$-fraction of the maximum utility value, where $\Gamma \geq 1$ is a problem-dependent constant.
Combining this with the submodular property implies that GADGET achieves an $\frac{1}{3\Gamma+1}$ overall competitive ratio.

\item Lastly, we conduct experiments to examine the performance of our GADGET algorithm.
We demonstrate the good approximation ratio of G-VNE, which is the main component of GADGET algorithm by using real-world trace-driven simulations.
We also show that, compared to existing baseline scheduling schemes, our GADGET algorithm can effectively meet the topology constraints, while achieving a good overall performance. 
\end{list}

Collectively, our results contribute to a comprehensive and fundamental understanding of RAR-based machine learning system optimization.
The rest of this paper is organized as follows.
We review related work in Section~\ref{sec:Related}, and present preliminaries to familiarize readers with the necessary background of RAR in Section~\ref{sec:Preliminaries}.
We then present the system model, problem formulation, and an overview of our algorithmic ideas in Section~\ref{sec:model_formulation}.
We propose our online resource scheduling algorithm by decoupling it into time-independent subproblems in Section~\ref{sec:alg}, and then solve the NP-hard subproblem with our G-VNE approach in Section~\ref{sec:VNE}.
We evaluate performance of our proposed algorithms through numerical experiments in Section~\ref{sec:numerical}, and conclude this paper in Section~\ref{sec:conclusion}.

%% file: Sec2-Related/Sec2-Related.tex

\section{Related Work} \label{sec:Related}
Due to the rise of deep learning and their intensive computation workload, scheduling optimization for DDL to expedite the training process has attracted increasing attention recently.
To date, the PS-worker architecture~\cite{Li14:parameter-server} has been widely adopted and its scheduling design has been relatively well studied (e.g., \cite{Chilimbi14:Adam,Yan15:MLFramework_KDD, Sun17:Dorm, Bao18:OASiS}).
However, as pointed out in Section~\ref{sec:intro}, the PS-worker architecture suffers from communication bottlenecks and reliability limitations.
Thanks to its better scalability compared to the PS-worker architecture~\cite{Sergeev18:Horovod}, the RAR architecture has received strong interest in the research community and has been recently adopted by modern DDL frameworks.
So far, however, results on scheduling designs for the RAR architecture remain scarce.
To our knowledge, PACE~\cite{Bao20:all-reduce} is the only existing work in the literature designed for all-reduce tensors 
based on the RAR architecture, aiming at maximizing the overlap between communication and computation using DAG (directed acyclic graph) of DNN training.
However, the goal of PACE is to speed up the training process of a {\em single job}, instead of optimizing the scheduling of multiple jobs to improve the system-wide performance (e.g., minimize the average completion time).
In contrast, our goal in this work is to design the {\em first} scheduler tailored for the RAR-based DDL jobs in computing clusters.
We propose a theoretical framework that enables rigorous RAR-based DDL training resource optimization in large-scale 
computing clusters (typically with a multi-layer hierarchical topology, e.g., fat-tree).

We note that there also exists other lines of research on resource scheduling for optimizing DDL training performance (e.g., latency, energy efficiency), but they are agnostic to the underlying distributed/parallel architectures.
Here, we also provide a quick overview, although they are not directly comparable to our work.
In~\cite{Gu19:Tiresias}, a GPU scheduler called Tiresias is proposed based on the assumption that DDL job performance could be estimated from historical job duration information.
Meanwhile, a resource provisioning method called Cynthia was developed in~\cite{Zheng19:Cynthia} to study the impact of system heterogeneity, i.e., servers in the system may have different hardware configuration (e.g., CPU, amounts of memory), on both synchronous and asynchronous training, and to optimize the training performance.
Also, due to the growing training workloads that incur huge energy consumption in the GPU clusters, the design of energy-efficient scheduling algorithms also receives significant interest recently. 
For example, in~\cite{Mei17:energy,Chau17:energy}, various scheduling schemes are proposed for CPU/GPU hybrid clusters, aiming at maximizing the energy efficiency without significantly sacrificing the system performance.
We note, however, that all of these DDL scheduling algorithms are heuristic methods that do not provide performance guarantees.

%% file: Sec3-Preliminaries/Sec3-Preliminaries.tex

\begin{figure}[t!]
\includegraphics[width=.5\textwidth]{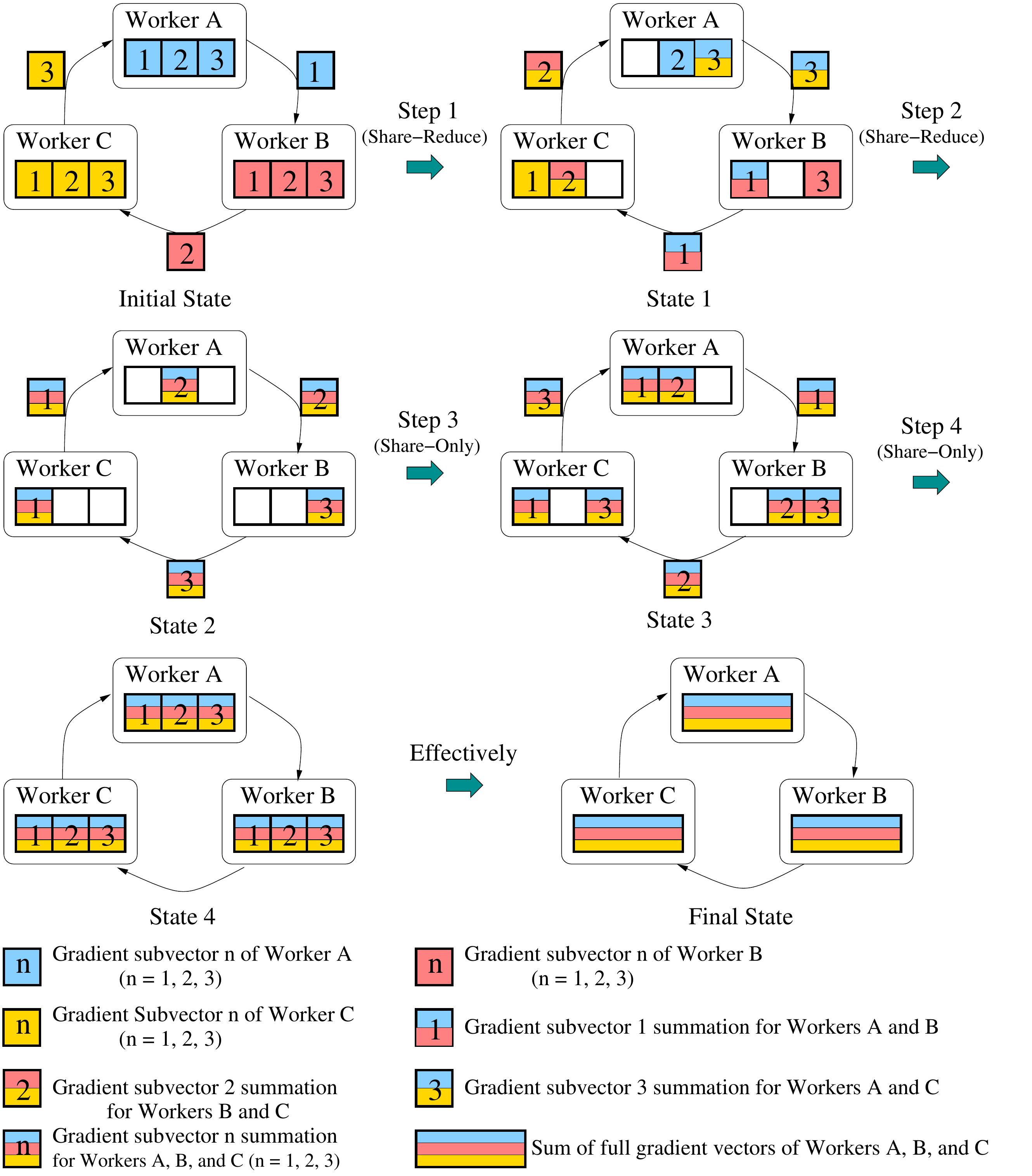}
\vspace{-.1in}
\caption{A three-worker illustrative example of the ring-all-reduce (RAR) process.}
\label{fig:ringallreduce}
\vspace{-.2in}
\end{figure}

\section{Distributed Learning with Ring-All-Reduce: A Primer}\label{sec:Preliminaries}

In this section, we  provide an overview on DDL training based on the RAR parallel computing architecture, to familiarize readers with necessary background and fix terminologies that will be used in this paper.

{\bf 1) Distributed Stochastic Gradient Descent (SGD):}
At the heart of deep learning lies an optimization problem in the form of $\min_{\w \in\mathbb{R}^d} \bar{L}(\w) \triangleq \frac{1}{N} \sum_{i=1}^{N} L(\w,\xi_i)$, where $\w$ contains the model parameters to be learned, $L(\w,\xi_i)$ is a loss function, 
and $N$ is the total number of samples.
In an $w$-worker DDL system, the dataset is often partitioned and trained by each worker.
Let $\mathcal{D}_j$ denote the $j$-th partition of the dataset trained by a worker $j\in \{1,\ldots,w\}$.
Then, the training problem can be decomposed as $\min_{\w \in\mathbb{R}^d} \bar{L}(\w) = \sum_{j=1}^{w} \frac{|\mathcal{D}_j|}{N} L_j(\w)$, where $L_j(\w) \triangleq \frac{1}{|\mathcal{D}_j|} \sum_{i \in \mathcal{D}_{j}} L(\w,\xi_i)$.
To date, most DDL systems in practice adopt the distributed stochastic gradient descent (SGD) method, where, in iteration $k$, the weight parameter vector $\w$ is updated as 
$\w_{k+1} = \w_{k} - \eta_k \big( \sum_{j=1}^{n} \frac{|\mathcal{D}_j|}{N} \g_j^k \big)$, $k=1,2,\ldots$. 
Here $\eta_k$ is the step-size (aka learning rate) in iteration $k$, and $\g_j^k$ denotes the stochastic gradient computed by worker $j$ in iteration $k$.

\smallskip
{\bf 2) The Ring-All-Reduce (RAR) Architecture:}
The above SGD update requires a weighted sum of all stochastic gradients, $\g_j$, $\forall j \in \{1,\ldots,w\}$ (the iteration index $k$ is omitted for notational simplicity).
To compute this weighted sum in the PS-worker architecture, each worker $j$ simply sends $\g_j$ to the PSs, which then perform the summation and return the result to each worker.
However, this implies that a $2wd$ amount of data exchange per iteration is required under the PS-worker architecture, which scales {\em linearly} as the number of workers increases and is problematic in large-scale DDL training.

To address this limitation, the more sophisticated RAR architecture is proposed.
The basic idea of RAR is to form a {\em ring} between the workers, where each worker performs gradient reduction (e.g., summation) and sharing by receiving gradients from its upstream worker and sending the local reduction result to its downstream worker.
In general, for a $w$-worker RAR structure, each worker splits its gradients into $w$ sub-vectors (see Fig.~\ref{fig:ringallreduce} for an example for $w=3$).
The RAR process is divided into two phases.
The first is the {\em ``Share-Reduce''} phase $(t=1,\ldots,w-1)$, where each worker sends its reduced sub-vector (i.e., the sub-vector sum) to its downstream worker, while receiving the reduced sub-vector from its upstream worker to compute a new reduced sub-vector.
The second phase is the {\em ``Share-Only''} phase  $(t=$ $w,\ldots,2w-2)$, where each node sends its newly received fully-reduced sub-vector to its downstream worker, while receiving a new fully-reduced sub-vector from its upstream worker.
Since each worker sends $\frac{d}{w}$ amount of data for $2(w-1)$ times, the total amount of data each worker sends is $\frac{2d(w-1)}{w}$, which is asymptotically {\em independent} of $w$ as $w$ gets large. 

\smallskip
{\bf 3) Per-Iteration Training Time for RAR-Based DDL:}
Consider a $w$-worker RAR-based DDL training job.
We use $b$ to denote the bandwidth between workers.
We use $G$ to denote the computation speed of a worker.
Since each worker sends a gradient sub-vector of size $d/w$ at each time step for $2(w-1)$ times, the transmission time can be computed as $\frac{2d(w-1)}{wb}$.
Also, since it takes $w-1$ times in total to perform gradient sub-vector summations, the computation time can be computed as $\frac{d(w-1)}{wG}$.
Hence, the total time for a single RAR operation is $\frac{d(w-1)}{w} [\frac{2}{b}+\frac{1}{G} ]$.
In each iteration of RAR-based DDL training, in addition to the all-reduce operation time, each worker needs to perform a forward pass (FP) and a backward pass (BP) to compute the stochastic gradients.
The FP time can be computed as $t^{f}M$, where $t^f$ is the model-dependent per-sample FP time and $M$ is mini-batch size.
The BP time $t^b$ is independent of the mini-batch size.
In addition, there is extra latency $\gamma$ caused by communication overhead (e.g., ACK time for message transmission, negotiation time among all workers before conducting RAR~\cite{Sergeev18:Horovod}). 
By putting all the above time consumption components together, we have the per-iteration training time $\tau$ for RAR-based DDL as:
\begin{align} \label{per-iteration-time}
\tau = \frac{d(w-1)}{w} \bigg[\frac{2}{b}+\frac{1}{G} \bigg]+t^fM+t^b + \gamma. \footnotemark
\end{align}
\footnotetext{In this paper, we focus on GPU computation and network communication times and neglect the data IO latency due to the use of high-speed solid state drives (SSDs), short pre-fetch time of pre-stored training programs, and overlapping of data chunks fetching and training through pipelining.}
We can see that $\tau$ depends on the learning model size, worker communication and computation speeds, batch size, communication overhead, and FP/BP times (in turn, the learning model), but is asymptotically upper bounded by $d(\frac{2}{b}+\frac{1}{G}) +t^fM+t^b + \gamma$ as the number of workers $w$ goes to infinity.

%% file: Sec4-Model/Sec4-Model.tex

\section{System Model and Problem Formulation} \label{sec:model_formulation}\label{sec:general_framework}

%
In this section, we will present a general analytical framework for resource scheduling and allocation for RAR-based DDL training performance optimization.
%


Consider a 
computing cluster with a set of physical servers denoted as $\mathcal{S}$.
Each server has a certain amount of computing resources (e.g., GPUs, memory).
The system is time-slotted with horizon $\mathcal{T} \triangleq \{1, 2, \cdots, T\}$.
Over time, RAR-based DDL jobs dynamically arrive and compete for the shared resources with unfinished jobs.
Let $a_{i}$ denote the arrival time of job $i$, which is {\em unpredictable} to the scheduler.
Let $\mathcal{I}$ denote the set of jobs arrived over time horizon $\mathcal{T}$.
The jobs are {\em preemptive} if resources are insufficient and could be resumed later.\footnote{Flexible resource allocation enabled by containers can be exploited to suspend a job and reclaim its resources without losing the execution progress~\cite{Chen17:USENIX}.}
All jobs' workers are implemented as containers. 
Next, we will develop RAR-based DDL scheduling models in detail.

\smallskip
{\bf 1) Resource Constraint Modeling:} 
We let $y_{is}[t]$ denote the number of workers scheduled on server $s$ for job $i$ in time-slot $t$.
Let $N_i$ be the largest number of assigned concurrent workers in each time-slot, we have: 
\begin{align}
 \label{ctr:maxworker}\sum_{s\in\mathcal{S}}y_{is}[t]\leq N_i, \quad\forall i\in\mathcal{I},  \forall t \in \mathcal{T}. 
\end{align}

Let $\mathcal{R}$ denote the set of computing resources (e.g., GPUs, memory, training time budget).
Let $l_i^r$ be the amount of type-$r$ resource consumed by each worker of job $i$.
We use $F_i^r$ to denote the maximum type-$r$ resource demand requested by job $i$.
To ensure that job $i$ does not exceed type-$r$ resource's limit, we have:
\begin{align}
    \label{ctr:iter}&\sum_{t\in\mathcal{T}}l_i^r \sum_{s\in\mathcal{S}}y_{is}[t]\leq F_i^r,\quad\forall i\in\mathcal{I}, \forall r \in \mathcal{R}.
\end{align}

Let $C_s^r$ denote the type-$r$ resource capacity of server $s$.
To ensure that server's type-$r$ limit is not violated, we have:
\begin{align}
    \label{ctr:cap}&\sum_{i\in\mathcal{I}} l_i^r y_{is}[t]\leq C_s^r, \quad\forall s\in\mathcal{S}, t\in\mathcal{T}, \forall r \in \mathcal{R}.
\end{align}

We use a binary variable $x_{is}[t]=1$ to indicate that job $i$ uses server $s$ in each time-slot $t$ and let $x_{is}[t]=0$ if otherwise.
Note that, when server $s$ is used for job $i$, $y_{is}[t]$ should not exceed any maximum resource demand $F_i^r$,  the resource capacity of server $s$, and the largest number of workers of job $i$ in each time-slot $t$.
Also, observe that $x_{is}[t]=0$ forces $y_{is}[t]=0$. 
Combining these facts yields:
\begin{align}
\label{ctr:placement-r}y_{is}[t]\!\leq\!x_{is}[t]\min \bigg\{N_i,\frac{C_s^r}{l_i^r}, \frac{F_i^r}{l_i^r}, \forall r \bigg\}, \forall i\!\in\!\mathcal{I},s\!\in\!\mathcal{S},t\!\in\!\mathcal{T}. \!\!\!
\end{align}
To ensure no workers are allocated before jobs arrive, we have:
\begin{align}
    &y_{is}[t]=0, \quad \forall i\in\mathcal{I}, s\in\mathcal{S}, t<a_i.
\end{align}

\smallskip
{\bf 2) RAR Topological Constraint Modeling:} 
A key component in scheduling an RAR-based DDL training job is to guarantee that the {\em physical network topology} corresponding to the resource scheduling decisions should be {\em compatible} with the {\em logical topology} of the job's computational graph.
Toward this end, we note that the computational graph of an RAR-based DDL training job $i$ is a {\em directed ring graph} $G_i=(V_i, E_i)$, where $V_i$ is the set of logical vertices representing workers, and $E_i$ represents the set of logical edges denoting the RAR directions between the workers.
Also, physical servers and network links in the cluster can be modeled as a directed substrate graph. 
Next, we model the topological constraints.

We use $\mathcal{P}_{ss'}[t]$ to denote the set of all possible paths between severs $s$ and $s'$ in the physical substrate graph in time-slot $t$.
We use a binary variable $r_{ss'}^{(p,i)}[t]=1$ to indicate that a path $p\in \mathcal{P}_{ss'}[t]$ is used by job $i$ in time-slot $t$, and let $r_{ss'}^{(p,i)}[t]=0$, otherwise.
If neither server $s$ nor $s'$ hosts any worker of job $i$ in time-slot $t$, then $r_{ss'}^{(p,i)}[t]=0$, which implies:
\begin{align}
\label{ctr:link}r_{ss'}^{(p,i)}[t]\leq x_{is}[t]x_{is'}[t], \forall s,s'\!\in\!\mathcal{S}, p\!\in\!\mathcal{P}_{ss'}[t], i\!\in\!\mathcal{I}, t\!\in\!\mathcal{T}.
\end{align}
Also, we use $b_i$ and $B_{\min}^p[t]$ to denote the reserved bandwidth requirement of job $i$ and the bottleneck link capacity of path $p \in \mathcal{P}_{ss'}[t]$, respectively.
To ensure that the bottleneck link capacity in any activated path $p$ is not exceeded, we have:
\begin{align}
\label{ctr:inter}&\sum_{i\in\mathcal{I}}r_{ss'}^{(p,i)}[t]b_i \leq B^p_{\min}[t], \forall s,s'\in\mathcal{S}, p\in \mathcal{P}_{ss'}[t], t\in\mathcal{T}.
\end{align}

  \begin{figure}[t!]
  \includegraphics[width=.52\textwidth]{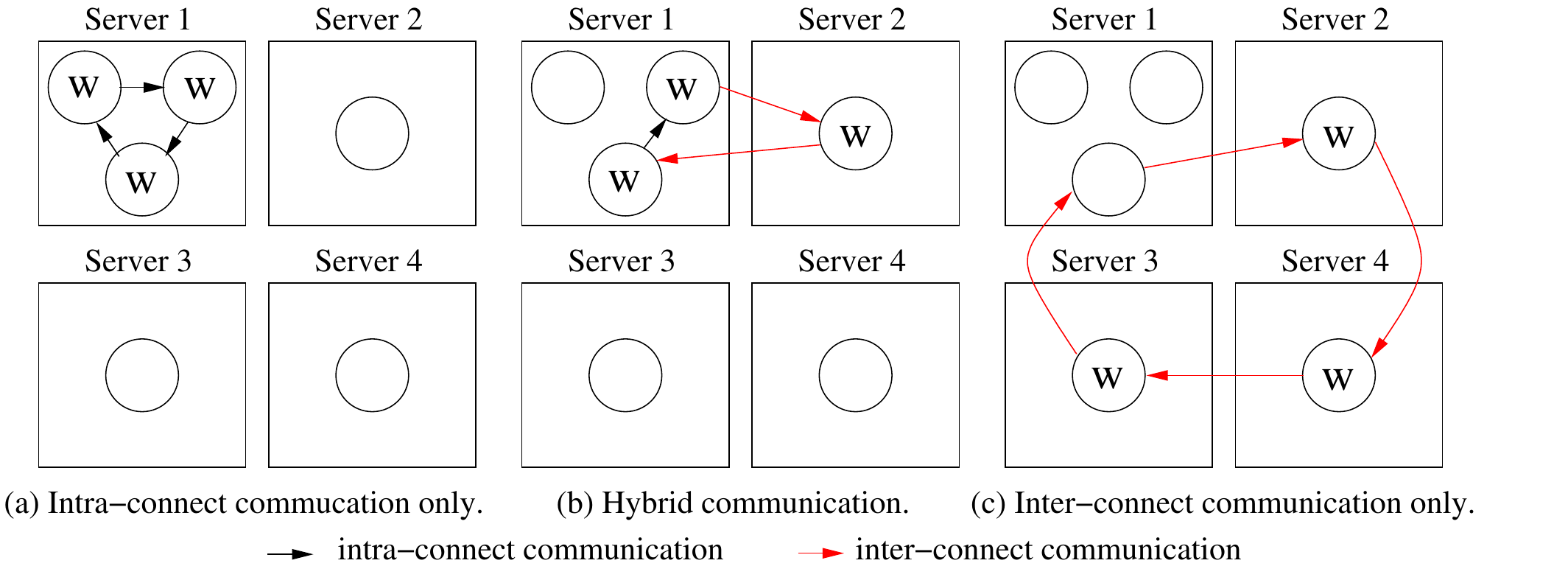}
  \caption{Various placement schemes with three workers.}
  \label{fig:placement}
  \vspace{-.2in}
\end{figure}

Next, to ensure that the allocated workers of job $i$ can indeed form a directed ring graph, we consider two cases.
First, if only one server hosts all workers for job $i$, then the cyclic constraint is automatically satisfied since a server can support any computational graph topology for co-located workers (Fig.~\ref{fig:placement}(a)).{\footnote{Servers in Fig.~\ref{fig:placement} are connected cyclically, i.e., there are communication links between servers 1 and 2, 2 and 4, 4 and 3, 3 and 1.}
Second, if workers are located on different servers, it can be observed that each server that hosts some of job $i$'s workers (i.e.,  $x_{is}[t]=1$) should have a degree of {\em exactly two} for job $i$'s paths in order to form a cycle (Fig.~\ref{fig:placement}(b)-(c)).
These two cases can be compactly written in one constraint as:
\begin{align} 
&\mathds{1}\bigg\{\sum_{s'\in\mathcal{S}}x_{is'}[t] \!-\!1\!>\!0 \bigg\} \bigg(\sum_{s' \in \mathcal{S}_i[t]} \sum_{p\in\mathcal{P}_{ss'}[t]}r_{ss'}^{(p,i)}[t]\!-\!2\bigg) = 0, \nonumber\\
  \label{ctr:cycle}&\hspace{1.5in} \quad\forall i\in\mathcal{I}, s\in\mathcal{S}, t\in\mathcal{T}.
\end{align}
where $\mathcal{S}_i[t]$ denotes the set of servers containing job $i$ at time $t$.
To see why Eq.~(\ref{ctr:cycle}) is a valid constraint, note that if all workers of job $i$ are hosted by server $s$, we have $\mathds{1} \{\sum_{s'\in\mathcal{S}}x_{is'}[t] - 1 > 0\} = 0$, which implies Eq.~\eqref{ctr:cycle} trivially holds.
Note also that the first summation in this term is over the set $\mathcal{S}_i[t]$.
This guarantees that there is only one big cycle instead of multiple small cycles.

\smallskip
{\bf 3) Objective Function and Problem Statement: } 
Let $\mu_{i}(\cdot)$ be the utility function associated with job $i$'s resource allocation, which is {\em non-decreasing} and concave to represent the ``diminishing return effect.''
In this paper, our goal is to maximize the overall utility of all jobs.
Let $\zeta_i$ represent some general notion of ``per-worker efficiency/cost'' (see three concrete examples next). 
Putting all modeling together, the DDL job scheduling problem (DDLJS) can be formulated as:
\begin{align*}
{\bf DDLJS: } \underset{\mathbf{x,y}}{\text{Maximize }} &\sum_{i\in\mathcal{I}}\mu_i \Big( \zeta_i \!\! \sum_{t\in\mathcal{T}} \sum_{s\in\mathcal{S}}y_{is}[t] \Big)\\
    \text{subject to } 
    &\text{Constraints }  (\ref{ctr:maxworker})-(\ref{ctr:cycle}).  
\end{align*}


We note that Problem~DDLJS is a general analytical framework that has many applications.
Here, we provide three examples to highlight its practical relevance:
{\bf 1) Excessive Training Avoidance~\cite{Jeon2018:multi-tenant}:} 
Here, $\zeta_i$ can represent the number of training iterations per-unit time of each worker of job $i$, which can be obtained by inverting Eq.~(\ref{per-iteration-time}).
The utility function can be chosen as $\mu(k) = (O(1/\sqrt{k}))^{-1} = C\sqrt{k}$ for some $C>0$, which is the typical convergence rate of SGD-type algorithms with respect to the iteration index $k$.
{\bf 2) Energy-Efficiency Optimization~\cite{Chau17:energy}:} 
Here, $\zeta_i$ denotes the per-worker power consumption of job $i$.
We can choose $\mu(\cdot) = -c(\zeta_i \sum_{t\in\mathcal{T}} \sum_{s\in\mathcal{S}}y_{is}[t])$, where $c(\cdot)$ is a quadratic function, which is often used in the power system literature to model energy consumption costs.
{\bf 3) Resource Fairness in Training:} Here, we can let $\zeta_i = 1$, $\forall i$, and adopt the classical ``proportional fairness'' utility function~\cite{Hou14:Proportional-Fair}, i.e., $\mu(\cdot) = \log(\cdot)$.

%% file: Sec5-Algorithm/Sec5-Algorithm.tex

\section{Solution Approach} \label{sec:alg}

Problem~DDLJS is a challenging {\em online} optimization problem (the scheduler does not have arrival information $\{a_i, \forall i\}$).
What exacerbates the problem is the fact that even its offline planning version (assuming $\{a_i, \forall i\}$ are known beforehand) is a mixed-integer non-convex programming (MINCP) problem, which is NP-Hard in general.
Moreover, Eq.~(\ref{ctr:cycle}) contains an indicator function that is not amenable to conventional optimization techniques.
In light of all these challenges, our goal in this paper is to pursue an online approximation algorithmic design that offers theoretical competitive ratio guarantee.

\subsection{Basic Idea}

To overcome the above challenges, we propose an online algorithmic design called GADGET (\underline{g}reedy ring-\underline{a}ll-reduce \underline{d}istributed \underline{g}raph \underline{e}mbedding \underline{t}echnique).
Our basic idea of GADGET contains two key steps:
i) Through a careful examination, we show that Problem~DDLJS is submodular with respect to scheduling decisions in the {\em temporal} domain.
Thus, it is possible to design a temporally greedy scheduling algorithm with a provable competitive ratio;
ii) For the resource allocation subproblem in each time-slot that remains NP-Hard, we show that it can be viewed as a {\em generalized} virtual network embedding problem (G-VNE).
As a result, it is also possible to (non-trivially) modify existing VNE techniques (see, e.g,. \cite{Rost18:VNE} and references therein) to adapt to our setting and solve each subproblem with provable approximation ratio guarantee.
We illustrate the basic idea of our GADGET algorithm in Fig.~\ref{fig:idea}.
In what follows, we will discuss these two key steps in detail.
\begin{figure}[t!]
  \includegraphics[width=1\columnwidth]{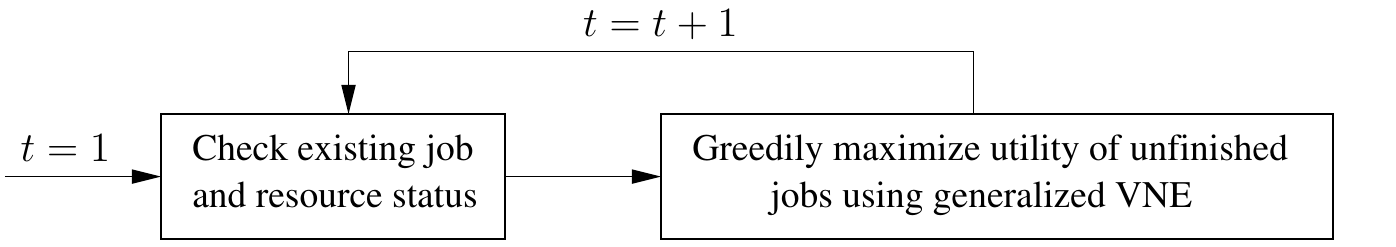}
  \caption{Algorithmic idea overview.}
  \label{fig:idea}
  \vspace{-.2in}
\end{figure}


\vspace{-.1in}
\subsection{An Online Temporally Greedy Approach}\label{subsec:greedy}
In this subsection, we establish the submodularity of Problem~DDLJS over the temporal domain.
First, consider the following online {\em temporally greedy} algorithm.
We let $z_{i,t}\triangleq\sum_{\tau=1}^t \sum_{s\in\mathcal{S}}y_{is}[\tau]$ be the {\em accumulative} number of worker-time allocated for job $i$ up to time $t$.
Define $\mathcal{I}[t]\triangleq\{i \in \mathcal{I}: t\geq a_i\text{ and } z_{i,t-1} < \min_{r\in\mathcal{R}} F_i^r/l_{i}^{r} \}$ as the set of jobs that are {\it active} (i.e., in-training and not resource-violated in time-slot $t$). 
Given the scheduling $\{z_{i,t-1}\}_{i\in\mathcal{I}[t]}$ in previous time slots, we greedily find the resource allocation in time-slot $t$ by solving the following optimization problem:
\begin{align}
\label{problem:pt}\text{ Maximize }
      &\sum_{i\in\mathcal{I}[t]}\mu_i \Big( \zeta_i\sum_{s\in\mathcal{S}}y_{is}[t]+\zeta_iz_{i,t-1} \Big)\\
    \text{subject to } 
    &\label{ctr:iter_t} \sum_{s\in\mathcal{S}}y_{is}[t]\leq \Big(\min_{r\in\mathcal{R}} \frac{F_i^r}{l_{i}^{r}} \Big)\!-\!z_{i,t-1}, \,\forall i\!\in\!\mathcal{I}[t],\\
    &\nonumber\text{Constraints }(\ref{ctr:maxworker}), (\ref{ctr:cap})-(\ref{ctr:cycle}) \text{ only at $t$ }, \forall i\!\in\!\mathcal{I}[t], \\ 
    &\nonumber y_{is}[t] \in \mathbb{Z}_{+}, \,\, x_{is}[t]\!\in\!\{0,1\},\forall i\!\in\!\mathcal{I}[t], s\in\mathcal{S}\nonumber.
\end{align}
Constraint~(\ref{ctr:iter_t}) ensures that the accumulated ``worker-time product'' in time slot $t$ does not exceed the remaining worker-time-product limit (determined by the bottleneck resource type) for all active jobs.
Note that Problem~(\ref{problem:pt}) remains a challenging NP-Hard packing problem, for which an approximation algorithm will be developed later in Section~\ref{sec:VNE}.
With Problem~(\ref{problem:pt}), our online temporally greedy algorithm is presented in Algorithm~\ref{alg:fractional_scheduling}.

\begin{algorithm}
\SetAlgoLined
\label{line:intialization}\textbf{Initialization:} Set $z_{i,t}\leftarrow0,\forall t\in\mathcal{T}, i\in\mathcal{I}$\;
 \label{line:iterate}\For{$t\in\mathcal{T}$}{
  \label{line:J[t]} $\mathcal{I}[t]\!\triangleq\!\{i\in\mathcal{I}: t\geq a_i\text{ and }z_{i}[t\!-\!1] \!\!< \!\min_{r\in\mathcal{R}} F_i^r/l_i^r \}$\;
  \label{line:subproblem1}$y_{is}[t]\leftarrow$ Solutions of Problem~(\ref{problem:pt}) using Algorithm~\ref{alg:round} developed in Section~\ref{sec:VNE}\;
  \Return $y_{is}[t]$\;
  \label{line:update_z}$z_{i,t} \leftarrow z_{i,t-1} + \sum_{s\in\mathcal{S}}y_{is}[t]$\;
  }
 \caption{Online Temporally Greedy Approach.}
 \label{alg:fractional_scheduling}
\end{algorithm}

Next, we show that Algorithm~\ref{alg:fractional_scheduling} provides competitive ratio guarantee by proving that Problem~DDLJS is submodular over the temporal domain.
For the paper to be self-contained, we restate some necessary basics of submodular optimization and matroid theory here, and refer readers to standard sources of submodular optimization (e.g.,~\cite{Schrijver03:combinatorial}) for further details.

\begin{defn}[Submodularity]
\label{def:3}A set function $f(\cdot): 2^\mathcal{V}\rightarrow\mathds{R}$ is submodular if for every $\mathcal{B}\subseteq\mathcal{V}$, and $\mathcal{A}'\subseteq\mathcal{A}\subseteq\mathcal{V}$, we have $f(\mathcal{A}\cup\mathcal{B}) - f(\mathcal{A})\leq f(\mathcal{A}'\cup\mathcal{B}) - f(\mathcal{A}')$.
\end{defn}
An important subclass of submodular functions are those that have the monotone property defined as follows:
\begin{defn}[Monotonicity]
A set function $f(\cdot): 2^\mathcal{V}\rightarrow\mathds{R}$ is monotone if for every $\mathcal{A}\subseteq\mathcal{B}\subseteq\mathcal{V}$, we have $f(\mathcal{A})\leq f(\mathcal{B})$.
\end{defn}
In this paper, we focus on non-negative monotone submodular functions.
Submodular optimization is also closely related to the notion of matroid, which is defined as follows:
\begin{defn}[Matroid]
A matroid is a pair $(\mathcal{V}, \mathcal{J})$ such that $\mathcal{V}$ is a finite set, and $\mathcal{J}\subseteq 2^\mathcal{V}$ is a collection of independent subsets  of $\mathcal{V}$ satisfying: 1) $\emptyset\in\mathcal{J}$; 2) for $\mathcal{A}\subseteq\mathcal{B}\subseteq\mathcal{V}$, if $\mathcal{B}\in\mathcal{J}$, then it implies $\mathcal{A}\in\mathcal{J}$; 3) if $\mathcal{A,B}\in\mathcal{J}$, and $|\mathcal{A}|<|\mathcal{B}|$, then $\exists \,\, v\in\mathcal{B}\setminus\mathcal{A}$ such that $\mathcal{A}\cup\{v\}\in\mathcal{J}$.
\end{defn}
Matroids have many interesting properties and subclasses.
One important subclass of matroids that is useful in this paper is the partition matroid, which is defined as follows:
\begin{defn}[Partition Matroid]
Partition $\mathcal{V}$ into disjoint subsets $\{\mathcal{V}_{j}\}$.
Let $0 \!\leq\! \nu_j \!\leq\! |\mathcal{V}_j|$, $\forall j$ be integers associated with the $\mathcal{V}_j$'s.
Define a collection of subsets $\mathcal{J} \!=\! \{ \mathcal{H} \!\in\! \mathcal{V}: |\mathcal{H} \!\cap\! \mathcal{V}_j| \!\leq\! \nu_j, \forall j\}$.
Then, the pair $(\mathcal{V}, \mathcal{J})$ is a partition matroid.
\end{defn}

We now show that Problem~DDLJS is a submodular optimization problem in the temporal domain.
To avoid ambiguity, we will use the term {\em ``schedule''} to refer to an $(\x,\y)$-decision over the entire time horizon $\mathcal{T}$;
and we use the term {\em ``allocation''} to refer to an $(\x,\y)$-decision in a particular time-slot (i.e., a ``snapshot'' in time).
We let $\mathcal{H} \triangleq \{ \mathcal{H}[t], t\in \mathcal{T} \}$ be the space of all  schedules, where $\mathcal{H}[t]$ denotes the space of all resource allocation in time-slot $t$ (which may or may not be feasible).
Define the ground set  $\mathcal{V} \triangleq \{y_{is}[t], \forall i,s,t | y_{is}[t]\in\mathcal{H}[t], \forall t \} \subset \mathbb{Z}_{+}^{|\mathcal{I}|\times |\mathcal{S}|\times|\mathcal{T}|}$ to be the space of all possible schedules of the $y$-components.
Also, let $\mathcal{V}[t] = \{y_{is}[t], \forall i, s | y_{is}[t]\in\mathcal{H}[t]\} \subset \mathbb{Z}_{+}^{|\mathcal{I}|\times |\mathcal{S}|}$ denote the feasible $y$-allocation space in time-slot $t$.
Clearly, $\{\mathcal{V}[t], t \in \mathcal{T}\}$ is a partition of $\mathcal{V}$ since $\mathcal{V}[t_1] \cap \mathcal{V}[t_2] = \emptyset$.
Now, we choose $\nu_t = 1$, $\forall t \in \mathcal{T}$.
Consider any schedule $\mathcal{E} \in \mathcal{V}$. 
Since at most one resource allocation decision can be chosen in each time slot in $\mathcal{E}$, we have $|\mathcal{E} \cap \mathcal{V}[t]| \leq1 = \nu_t$, $\forall t$.
Let $\mathcal{J}$ be the collection of all such feasible schedules $\{\mathcal{E}\}$.
Then, the pair $(\mathcal{V},\mathcal{J})$ containing all schedules forms a partition matroid, and finding an optimal feasible schedule is equivalent to finding an optimal independent set of this partition matroid.

Next, we let $z_i[t] = \sum_{s\in\mathcal{S}}y_{is}[t]$ denote the total number of workers of job $i$ in time-slot $t$.
Clearly, $z_i[t]=0, \forall t < a_i$.
Recalling the definition of $z_{i,t}$, we have $z_{i,t}=z_i[a_i]+z_i[a_i + 1] + \cdots+z_i[t]$.
For any schedule $\mathcal{E}$ (could be infeasible), we define $z^{\mathcal{E}}_{i}[t] \triangleq \sum_{\{ y_{is}[t] \in \mathcal{E} \cap \mathcal{V}[t] \}} y_{is}[t] $.
Accordingly, we define $z_{i,t}^\mathcal{E} \triangleq z_i^\mathcal{E}[a_i]+\cdots+z_i^\mathcal{E}[t]$.
Then, we can rewrite the objective function in (\ref{problem:pt}) as $F(\mathcal{E}) \triangleq \sum_{i\in\mathcal{I}} \mu_{i}(\zeta_i z_{i,T}^{\mathcal{E}})$.
Since $\mu_{i}(\cdot)$ is concave and increasing, $F(\mathcal{E})$ is monotonically increasing as $t$ increases.
Next, we show that $F(\mathcal{E})$ is also submodular.
\begin{lem}
$F(\mathcal{E})$ is a submodular function.
\end{lem}

\begin{proof}
We pick two schedules $\mathcal{A}$ and $\mathcal{B}$ such that $\mathcal{A}\subseteq\mathcal{B}\subseteq\mathcal{V}$.
Now, for a resource allocation decision $v \in \mathcal{V}[t]$ in time-slot $t$.
Consider the following two cases:
\begin{list}{\labelitemi}{\leftmargin=1em \itemindent=0.em \itemsep=.2em}
\item Case 1) $v\in\mathcal{B}$:
In this case, we have $F(\mathcal{A}\cup\{v\}) - F(\mathcal{A}) \geq F(\mathcal{B}\cup\{v\}) - F(\mathcal{B}) =0$, which holds trivially.

\item Case 2) $v\not\in\mathcal{B}$: In this case, we have
\begin{align*}
&\!\!\!\! F(\mathcal{B} \!\cup\! \{v\}) \!-\! F(\mathcal{B}) \!=\!\! \sum_{i\in\mathcal{I}}\big(\mu_i(\zeta_i z_{i,T}^\mathcal{B} \!+\! \zeta_iz_i^v[t]) \!-\! \mu_i(\zeta_iz_{i,T}^\mathcal{B})\big) \\
 &\!\!\!\!\!\! \overset{(a)}{\leq} \!\!\sum_{i\in\mathcal{I}} \big(\mu_i(\zeta_i z_{i,T}^\mathcal{A} \!+\! \zeta_i z_i^v[t]) \!-\! \mu_i(\zeta_i z_{i,T}^\mathcal{A}) \big) \!\!=\!\! F(\mathcal{A} \!\cup\! \{v\}) \!-\! F(\mathcal{A}),
\end{align*}
\end{list}
where $(a)$ follows from $\mu_i(\cdot)$ being concave, monotone, and increasing.
Then, the proof is complete by Definition~\ref{def:3}.
\end{proof}

The main competitive ratio result of GADGET is stated in the following theorem.
Due to space limitation, we provide a proof sketch in this paper.
\begin{thm}\label{thm:greedy}
Algorithm~\ref{alg:fractional_scheduling} produces a schedule that is $\frac{\alpha}{\alpha+1}$- competitive, where $\alpha \in (0,1]$ is the approximation ratio of solving Problem~(\ref{problem:pt}) in each time-slot.
\end{thm}

\begin{proof}[Proof Sketch]
The main idea of our proof is to leverage the $p$-system result for greedy algorithms with approximation, where $p$ denotes the ratio between the maximum and minimum cardinalities of maximal independent sets in a family of independent sets.
It has been shown (e.g., \cite{Calinescu11:matroid}) that applying an online greedy algorithm with an approximation ratio of $\alpha \!\in\! (0,1]$ in each round for a $p$-system yields a competitive ratio of $\frac{\alpha}{\alpha+p}$.
Hence, the result in Theorem~\ref{thm:greedy} is proved once we show that Problem~DDLJS is a 1-system ($p\!=\!1$).

Toward this end, recall that the space of all possible schedules $\mathcal{V}=\{y_{is}[t], \forall i, s, t\}$ is partitioned into a collection of disjoint ``allocations'' $\{\mathcal{V}[t]\}$, $t\in\mathcal{T}$, where $\mathcal{V}[t]=\{y_{is}[t], \forall i, s\}$ and the independence family is $\mathcal{J} = \{ \mathcal{H} \in \mathcal{V}: |\mathcal{H} \!\cap\! \mathcal{V}[t]| \!\leq\! 1, \forall t\}$.
For any $Y \!=\! \{\mathcal{U}[t] \!\subseteq\! \mathcal{V}[t], \forall t \} \!\subseteq\! \mathcal{V}$, let $\mathcal{B}(Y)$ be the set of maximal independent sets of $\mathcal{V}$ included in $Y$,
which implies that $|\mathcal{H}\cap\mathcal{U}[t]| \leq 1$.
Then, we can show that $\mathcal{B}(Y)=\{\{u_k[t]\in\mathcal{U}[t], \forall t\},\forall k\}$, where $k$ is the element (scheduling) index in the set $\mathcal{U}[t]$.
Thus for any $A\in\mathcal{B}(Y)$, the cardinality of $A$ is the same (i.e., $|A|=|Y|$), which immediately implies that $p=1$ and the proof is complete.
\end{proof}

\subsection{Solving the Per-Time-Slot Problem in (\ref{problem:pt})} \label{sec:VNE}

Under the temporally greedy approach in Algorithm~\ref{alg:fractional_scheduling}, it remains to solve an NP-Hard problem in (\ref{problem:pt}).
Due to the challenges in directly handling the ring topological constraint in \eqref{ctr:cycle} that is mixed-integer and highly unstructured,
we take an ``indirect'' approach by noting that Problem~\eqref{problem:pt} is a {\em generalized} virtual network embedding (VNE) problem (embedding a virtual computational graph onto a substrate physical graph while respecting all capacity constraints).
Notably, VNE with request graphs of cactus-type topologies has been solved in~\cite{Rost18:VNE}, which includes the ring topology as a special case.
The basic idea of the solution is based on randomized rounding the routing graphs over the underlying multi-commodity flow (MCF) problem (see, e.g., \cite{Rost18:VNE,Raghavan85:graphrandomizedrounding} for details).

However, our work differs from standard VNE~\cite{Rost18:VNE} in the following two key aspects: 
1) Unlike in standard VNE where the number of nodes in each request graph is given, the  number of nodes (i.e., workers) is part of the problem in \eqref{problem:pt}; 
2) Unlike standard VNE with a one-dimensional resource capacity constraint, Problem~\eqref{problem:pt} has  {\em multi-dimensional} resource capacity constraints.
Due to these differences, we refer to Problem~\eqref{problem:pt} in this paper as the {\em generalized} VNE (G-VNE).

To address these challenges, we again take a {\em ``divide-and-conquer''} approach:
i) We observe that, although the number of workers is unknown, its upper bound $q_i[t]$ can be obtained efficiently by solving Problem~(\ref{problem:pt}) with continuous relaxations of Constraints (\ref{ctr:maxworker}), (\ref{ctr:cap}), and (\ref{ctr:iter_t});
ii) Once the upper bound of the number of workers is known, we can reformulate the G-VNE problem with a ``one-hot'' worker number constraint.

To reformulate the G-VNE problem, we let $\rho_i[t] \!\in\!\{0,1\}$ be the binary variable to indicate whether job $i$ is embedded in time slot $t$.
Let $\mathcal{Q}_i[t] \triangleq \{1, \ldots, \lceil q_i[t] \rceil \}$ be the set of all possible numbers of workers at time t for job $i$.
Let binary variable $\chi_{i,\kappa}[t] \in \{0,1\}$ denote whether a ring of size $\kappa \in\mathcal{Q}_i[t]$ is chosen at time $t$.
We use $V_i^\kappa$ and $E_i^\kappa$ to represent the sets of nodes and edges of the chosen $\kappa$-ring, respectively.
Let $\varrho^u_{i, \kappa,a}[t]$ and $\theta^{u,v}_{i,\kappa,a,b}[t]$ denote whether node $a\in V^\kappa_i$ and edge $(a,b)\!\in\!E^\kappa_i$ are mapped to a physical node $u\!\in\!V_s$ and a physical link $(u, v)\in E_s$ at time $t$, respectively.

Next, we use $d_{i,\kappa}^{a, r}$ and $d_{i,\kappa}^{a,b}$ to denote the type-$r$ resource demand for node $a\in V_i^{\kappa}$ and resource demand for edge $(a,b)\in E_i^{\kappa}$ in the chosen $\kappa$-ring of job $i$, respectively.
Let $g_s^{u,r}$ and $e_s^{u,v}$ be the type-$r$ resource capacities of node $u\in V_s$ and edge $(u,v)\in E_s$ in the substrate network, respectively.
Also, let $h_i^{u,r}[t]$ and $o_i^{u, v}[t]$ denote the cumulative type-$r$ resource allocation on node $u$ and the cumulative and aggregated resource allocation on edge $(u,v)$ of job $i$ in the substrate network at time $t$, respectively.
Let $\pi_{i,\kappa}[t]$ be the incremental utility with the $\kappa$-ring chosen at time $t$, i.e., $\pi_{i,\kappa}[t]=\mu_i(\zeta_iz_{i,t-1}+\zeta_i \sum_{\kappa \in \mathcal{Q}_i[t]} \kappa \chi_{i,\kappa}[t]  -\mu_i(\zeta_iz_{i,t-1})$.
We let $\delta^+(u)$ and $\delta^-(u)$ denote the sets of outgoing and incoming edges of a node $u$, respectively.
Then, the G-VNE problem in time-slot $t$ can be reformulated as an integer linear program (ILP) (omitting time index ``$[t]$''  for notational simplicity): 
\begin{align}
&\label{problem:mcf}\text{Maximize }
     \sum_{i\in\mathcal{I}}\sum_{\kappa\in\mathcal{Q}_i}\pi_{i,\kappa}\chi_{i,\kappa}\\
    &\label{lp:1}\sum_{\kappa\in\mathcal{Q}_i}\chi_{i,\kappa} = \rho_i,\quad\forall i\in\mathcal{I},\\
    &\label{lp:2}\sum_{u\in V_s}\varrho^u_{i,\kappa,a}= \chi_{i,\kappa}, \quad \forall i\in\mathcal{I}, \kappa\in\mathcal{Q}_i, a\in V_i^\kappa, \\
    &\nonumber\label{lp:3}\sum_{\substack{(u,v) \in\delta^+(u)}} \theta^{u,v}_{i,\kappa, a, b} - \sum_{\substack{(v,u) \in\delta^-(u)}}\theta^{v,u}_{i,\kappa,a,b}= \varrho^u_{i,\kappa, a} -\varrho^u_{i,\kappa, b}, \\
    &\hspace{1in}\forall i\in\mathcal{I}, \kappa\in\mathcal{Q}_i, (a,b)\in E_i^\kappa, u\in V_s,\\
   &\label{lp:4} \sum_{\kappa\in\mathcal{Q}_i}\sum_{a\in V^\kappa_i}d_{i,\kappa}^{a,r} \varrho^u_{i,\kappa,a} = h_i^{u,r}, \quad \forall i\in\mathcal{I}, u\in V_s, r\in\mathcal{R},\\  
    &\label{lp:5}\sum_{\kappa\in\mathcal{Q}_i}\sum_{(a,b)\in E^\kappa_i}d_{i,\kappa}^{a,b} \theta^{u,v}_{i,\kappa,a,b} = o_i^{u,v}, \forall i\in\mathcal{I}, (u,v)\in E_s,\\ 
    &\label{lp:6}\sum_{i\in\mathcal{I}}h_i^{u,r}\leq g_s^{u,r}, \quad\forall u \in V_s,\forall r\in\mathcal{R},\\
    &\label{lp:7}\sum_{i\in\mathcal{I}}o_i^{u,v}\leq e_s^{u,v}, \quad\forall (u,v) \in E_s.
\end{align}
Here, Constraint~(\ref{lp:1}) ensures that at most one of the $\kappa$-ring graphs ($\kappa=1,\ldots,\lceil q_i \rceil $) from $\mathcal{Q}_i$ can be selected.
Constraint~(\ref{lp:2}) ensures that if job $i$ is embedded ($\rho_i=1$) and if a $\kappa$-ring graph is selected ($\chi_{i,\kappa}=1$), then each node $a\in V_i^\kappa$ must be mapped to some physical node in $V_s$, i.e., there exists a node $u\in V_s$ such that $\varrho_{i,\kappa,a}^u = 1$.
Constraint~(\ref{lp:3}) induces a non-splittable unit flow for each edge $(a,b)\in E_i^\kappa$ from the physical location that $a$ is mapped to the physical location that $b$ is mapped.
Constraints~(\ref{lp:4})--(\ref{lp:5}) compute the cumulative resource on physical nodes and edges, respectively.
Constraints~(\ref{lp:6})--(\ref{lp:7}) ensure no violation of the physical nodes and edges resource capacities, respectively.
%
Note that Constraints~(\ref{lp:1}), (\ref{lp:4}) and (\ref{lp:6}) are the {\em major differences} from the MCF formulation of the standard VNE problem (see, e.g., \cite{Rost18:VNE}), where each job has only one request graph (i.e., $|\mathcal{Q}_i|=1$) and each node has only one resource type.

Now, consider a linear program (LP)-based ring selection approach:
Let $\{\bar{\chi}_{i,\kappa}, \bar{\rho}_i\}$ be the LP-relaxation solution of Problem~(\ref{problem:mcf}) and $\bar{\pi}_{i,k}$ be the its utility value.
Select $\kappa_i = \max_{\kappa\in\mathcal{Q}_i : \bar{\chi}_{i,\kappa>0}}\{\bar{\pi}_{i,\kappa}\bar{\chi}_{i,\kappa}\}$ as the request ring-graph of job $i$ for $\bar{\rho}_i>0$.
That is, we set $\chi_{i,\kappa_i}\!=\!1$ and $\chi_{i,\kappa}\!=\!0$ if $\kappa\!\neq\!\kappa_i$, for each job with $\bar{\rho}_i>0$.
Then, we can show that this LP-based ring selection scheme has the following approximation ratio:

\begin{lem}[LP-based Ring Selection]\label{lem:selection}
Let $\Gamma\!\!=\!\!\max_i \lceil q_i\rceil$.
The LP-based ring selection scheme achieves at least  a $1/\Gamma$-fraction of the utility obtained by an offline optimal approach.
\end{lem}
\begin{proof}
Let $\chi^*_{i,\kappa^*}$, $\pi^*_{i,\kappa^*}$ be the optimal solution and objective value with $\kappa^*$ being the optimal ring size, respectively.
Let $\hat{\chi}_{i,\hat{\kappa}}$, $\hat{\pi}_{i,\hat{\kappa}}$ be the solution and objective value of our LP-based scheme, respectively, with $\hat{\kappa}$ being the ring size of the LP-based scheme.
Let $\Pi_i\triangleq\min\{N_i,\frac{C_s^r}{l_i^r},\frac{F_i^r}{l_i^r}, \forall r,s \}, \forall i$.
Then, 
\begin{align*}
&\frac{\sum_i\pi^*_{i,\kappa^*}\chi^*_{i,\kappa^*}}{\sum_i\hat{\pi}_{i,\hat{\kappa}}\hat{\chi}_{i,\hat{\kappa}}}\leq\frac{\sum_i\sum_\kappa\bar{\pi}_{i,\kappa}\bar{\chi}_{i,\kappa}}{\sum_i\hat{\pi}_{i,\hat{\kappa}}\hat{\chi}_{i,\hat{\kappa}}}\overset{(a)}{\leq}\max_i\frac{\sum_\kappa\bar{\pi}_{i,\kappa}\bar{\chi}_{i,\kappa}}{\hat{\pi}_{i,\hat{\kappa}}\hat{\chi}_{i,\hat{\kappa}}}\overset{(b)}{\leq}\Gamma.
\end{align*}
To see why $(a)$ holds, let $\bar{\iota}_i\triangleq\sum_\kappa\bar{\pi}_{i,\kappa}\bar{\chi}_{i,\kappa}$ and $\hat{\iota}_i\triangleq \hat{\pi}_{i,\hat{\kappa}}\hat{\chi}_{i,\hat{\kappa}}$.
Let $i^* = \arg\max_{i} \big\{ \frac{\sum_\kappa\bar{\pi}_{i,\kappa}\bar{\chi}_{i,\kappa}}{\hat{\pi}_{i,\hat{\kappa}}\hat{\chi}_{i,\hat{\kappa}}} \big\}$.
Thus, we have
$\frac{\bar{\iota}_{i^*}}{\hat{\iota}_{i^*}} - \frac{\sum_i\bar{\iota}_i}{\sum_i\hat{\iota}_i}=\frac{\bar{\iota}_{i^*}\sum_i\hat{\iota}_i-\hat{\iota}_{i^*}\sum_i\bar{\iota}_i}{\hat{\iota}_{i^*}\sum_i\hat{\iota}_i} \!=\! \frac{\sum_i(\bar{\iota}_{i^*}\hat{\iota}_i-\bar{\iota}_i\hat{\iota}_{i^*})}{\hat{\iota}_{i^*}\sum_i\hat{\iota}_i} \!\geq\! 0$. 
Also, $(b)$ follows from i) the LP relaxation: $\frac{\bar{\pi}_{i,\kappa}\bar{\chi}_{i,\kappa}}{\hat{\pi}_{i,\hat{\kappa}}\hat{\chi}_{i,\hat{\kappa}}} \!\leq\! 1, \forall i, \forall \kappa \!\in\! \mathcal{Q}_i$); and ii) $|\mathcal{Q}_i| \!=\! \lceil q_i\rceil \!\leq\! \Pi_i$.
This completes the proof.
\end{proof}

Next, upon determining the ring size $\kappa$, we perform virtual network embedding with multi-dimensional resource constraints.
First, similar to \cite[Formulation~2]{Rost18:VNE}, we resolve the uncertainty that may occur in embedding cyclic graphs by creating an augmented LP \cite{Rost18:VNE} that binds multiple copies of a ``reduced version'' of Problem~(\ref{problem:mcf}), where the set of $\hat{\kappa}$-rings has been chosen following the LP-based ring selection.
Then, we solve the augmented LP to obtain a relaxation solution, which can be used to recover a set of mapping-selection tuples $\mathcal{M}_i=\{(\varphi_i^k, \omega_i^k), k=1,\ldots,|\mathcal{M}_i| : \varphi_i^k>0, \sum_{k} \varphi_i^k \leq1 \}$ for each job $i$, where the mapping $\omega_i^k$ (a candidate of embedding) is chosen with probability $\varphi_i^k$, and rejected with probability $1-\sum_k\varphi_i^k$.
Note that to address the {\em multi-dimensional resource} challenge in computing $\mathcal{M}_i$, our key idea is to conduct mapping search to determine $\omega_i^k(r)$ for each $r \in \mathcal{R}$ by leveraging techniques in \cite[Sec.~III-C]{Rost18:VNE}.
Then, we choose $\omega_{i}^{k}\!=\!\cap_{r\in\mathcal{R}}\omega_{i}^{k}(r)$.
After the set $\{ \mathcal{M}_i, \forall i \}$ is calculated, we perform randomized rounding with probabilities based on $\varphi_i^k$-values  to obtain the embedding.
Putting all these together, we summarize our LP-based ring-selection and multi-dimensional resource embedding (LP-RS-MDE) method in Algorithm~\ref{alg:round}.
\begin{algorithm}
\SetAlgoLined
\label{line:intialization}\textbf{Initialization:} Choose $u_b$ as the maximum number of roundings, and set $iter\leftarrow 1$\;
Set the values of constants $\alpha$, $\beta^r$, $\gamma$  (see Theorem~\ref{thm:ratio})\;
Solve the LP relaxation of Problem~(\ref{problem:mcf}) and set ring sizes $\kappa_i = \max_{\kappa\in\mathcal{Q}_i : \bar{\chi}_{i,\kappa>0}}\{\bar{\pi}_{i,\kappa}\bar{\chi}_{i,\kappa}\}$ for all job $i$\label{line:LP1}\;
Solve an augmented LP  with ring sizes fixed in Step~\ref{line:LP1}\label{line:LP2}\;
\For{$i\in\mathcal{I}[t]$}{\label{line:dmp_start}
	Compute $\mathcal{M}_i\!\!=\!\!\{(\varphi_i^k,\omega_i^k)\}$ based on the solution of the augmented LP, where $\varphi_i^k, \omega_i^k$ are determined by \cite[Sec.~III-C]{Rost18:VNE} and the bottleneck resource\label{line:dmp_end}\;
}
\While{solution is not $(\alpha, \beta^r, \gamma)$-approx. \& $iter\!<\!u_b$}{\label{line:round_start}
	{\bf foreach} $i\in\mathcal{I}[t]$ choose $\omega_i^k$ with probability $\varphi_i^k$\;
	$iter\leftarrow iter + 1$\label{line:round_end}\;
	
}
 \caption{LP-based Ring-Selection and Multi-Dimensional Resource Embedding (LP-RS-MDE).}
 \label{alg:round}
\end{algorithm}

To analyze LP-RS-MDE's performance, we let $d^r_{max}(i, u)$ and $d_{max}(i, u, v)$ denote the maximal type-$r$ resource demands that job $i$ imposes on node $u$ and edge $(u,v)$ in the substrate network, respectively.
Let $C^r_{max}(i, u)$ and $C_{max}(i, u, v)$ denote the maximal type-$r$ resource allocation that a valid mapping of job $i$ may impose on the substrate network's nodes and edges, respectively.
We define two constants as follows:
\begin{align*}
&\Delta^r(V_s) \triangleq \max_{u}\sum_{i\in\mathcal{I}[t]:d^r_{max(i,u)>0}}\!\!\!\!(C^r_{max}(i,u)/d^r_{max}(i,u))^2, \forall r,\\
&\Delta(E_s) \triangleq \max_{(u,v)}\sum_{i\in\mathcal{I}[t]:d_{max(i,u,v)>0}}\!\!\!\!\!\!(C_{max}(i,u,v)/d_{max}(i,u,v))^2.
\end{align*}

Also, we let $A^{u,r}[t]$ and $A^{u,v}[t]$ represent the overall type-$r$ resource allocation on node $u\in V_s$ and on edge $(u,v)\in E_s$ after randomized rounding at time $t$, respectively.
Then we have following analytical results:
\begin{thm}[Performance of LP-RS-MDE]
\label{thm:ratio} Assume that the substrate network has at least three servers ($|V_s|\geq 3$).
Then, LP-RS-MDE achieves at least 
$(1/3)$-fraction of the optimal value of Problem~\eqref{problem:mcf}, with probabilities of resource constraint violations satisfying $\mathbb{P}\{A^{u,r}[t] \geq \beta^rg_s^{u,r}\} \leq |V_s|^{-4}, \forall r, u, t$, and $\mathbb{P}\{A^{u,v}[t]\geq \gamma e_s^{u,v}\}\leq |E_s|^{-4}, \forall (u,v), t$, where $\beta^r =1+\epsilon\sqrt {2\Delta^r(V_s)\log(|V_s|)}$, $\forall r$, and $\gamma = 1 + \epsilon\sqrt{2\Delta(E_s)\log(|E_s|)}$.
\end{thm}

To prove Theorem~\ref{thm:ratio}, note that unlike the single-dimensional resource capacity in standard VNE~\cite{Rost18:VNE}, the analysis of Algorithm~\ref{alg:round} needs to consider multi-dimensional resource capacity violation (see Eqs.~(\ref{lp:4}) and (\ref{lp:6})).
Toward this end, we observe that the mapping computation in Step~\ref{line:dmp_end} of Algorithm~\ref{alg:round} implies that a similar approach as in \cite{Rost18:VNE} can be used to analyze the final mapping for each type-$r$ resource in each job, which in turn leads to the results stated in Theorem~\ref{thm:ratio}.
We omit the proof details here due to space limitation.
Next, we analyze the running time complexity of LP-RS-MDE.

\begin{thm}[Time Complexity of LP-RS-MDE]\label{thm:time}
LP-RS-MDE has polynomial time complexity $O((|\mathcal{I}[t]|\cdot|G_s|)^3 + u_b\cdot\mathcal{I}[t])$.
\end{thm}

\begin{proof}
The main components of the running time include solving the LP relaxations to obtain fractional solutions (Lines~\ref{line:LP1}-\ref{line:LP2} in Algorithm~\ref{alg:round}), computing valid mappings using decomposition approach (Lines~\ref{line:dmp_start}-\ref{line:dmp_end} in Algorithm~\ref{alg:round}), and then performing randomized rounding on these computed mappings (Lines~\ref{line:round_start}-\ref{line:round_end} in Algorithm~\ref{alg:round}). 
The time complexity of solving the LP relaxations is $O((|\mathcal{I}[t]| \cdot |G_s|)^3)$~\cite{VaidyaLP:90}.
The mapping search step terminates when $\varrho_i \leq 0$, and in each termination, at least one variable becomes $0$.
The number of variables for each job $i$ is bound by $O(|G_i| \cdot |G_s|)$ for each copy of Problem~(\ref{problem:mcf}).
Specifically, the mapping search visits all nodes and edges in the request graph.
Also, in each visit, in order to find the set of mapping-selection tuples $\{ \mathcal{M}_i$, $\forall i\}$, the substrate network will be traversed.
Thus, the total number of variables is upper bounded by $O(|G_i| \cdot |G_s|)$.
Further, there are at most $O(|V_s|)$ copies of Problem~\eqref{problem:mcf}, which implies $O(|G_i|\cdot|G_s|\cdot |V_s|\cdot|\mathcal{I}[t]|)$ complexity.
Finally, the rounding time of each job $i$ is bounded by $u_b$, and thus the running time is bounded by $O(u_b\cdot\mathcal{I}[t])$. 
Hence, the overall time complexity of the LP-RS-MDE method is $O((|\mathcal{I}[t]|\cdot|G_s|)^3\!+\!u_b\cdot\mathcal{I}[t])$, and the proof is complete.
\end{proof}

By combining results in Theorems~\ref{thm:greedy} and \ref{thm:ratio}, we have: 

\begin{thm}[Competitive Ratio of GADGET]\label{thm:overallratio}
Algorithm~\ref{alg:fractional_scheduling} produces a schedule that yields a utility value at least  $(\frac{1}{3\Gamma+1})$-fraction of the maximum utility value of Problem~DDLJS, with probabilities of resource violations satisfying $\mathbb{P}\{A^{u, r} [t] \geq \beta^r g_s^{u,r}, \exists t \in \mathcal{T} \} \leq 1-(1-|V_s|^{-4})^T, \forall u, r$, and $\mathbb{P}\{A^{u, v} [t] \geq \gamma e_s^{u,v}, \exists t \in \mathcal{T} \} \leq 1-(1-|E_s|^{-4})^T, \forall (u,v)$, where $\Gamma$ and $\beta^r$, $\gamma$ are as defined in Lemma~\ref{lem:selection} and Theorem~\ref{thm:ratio}.
\end{thm}

%% file: Sec6-Numerical/Sec6-Numerical.tex

\begin{figure*}[t!]
 \begin{minipage}[t]{0.24\linewidth}
        \includegraphics[width=1.1\textwidth]{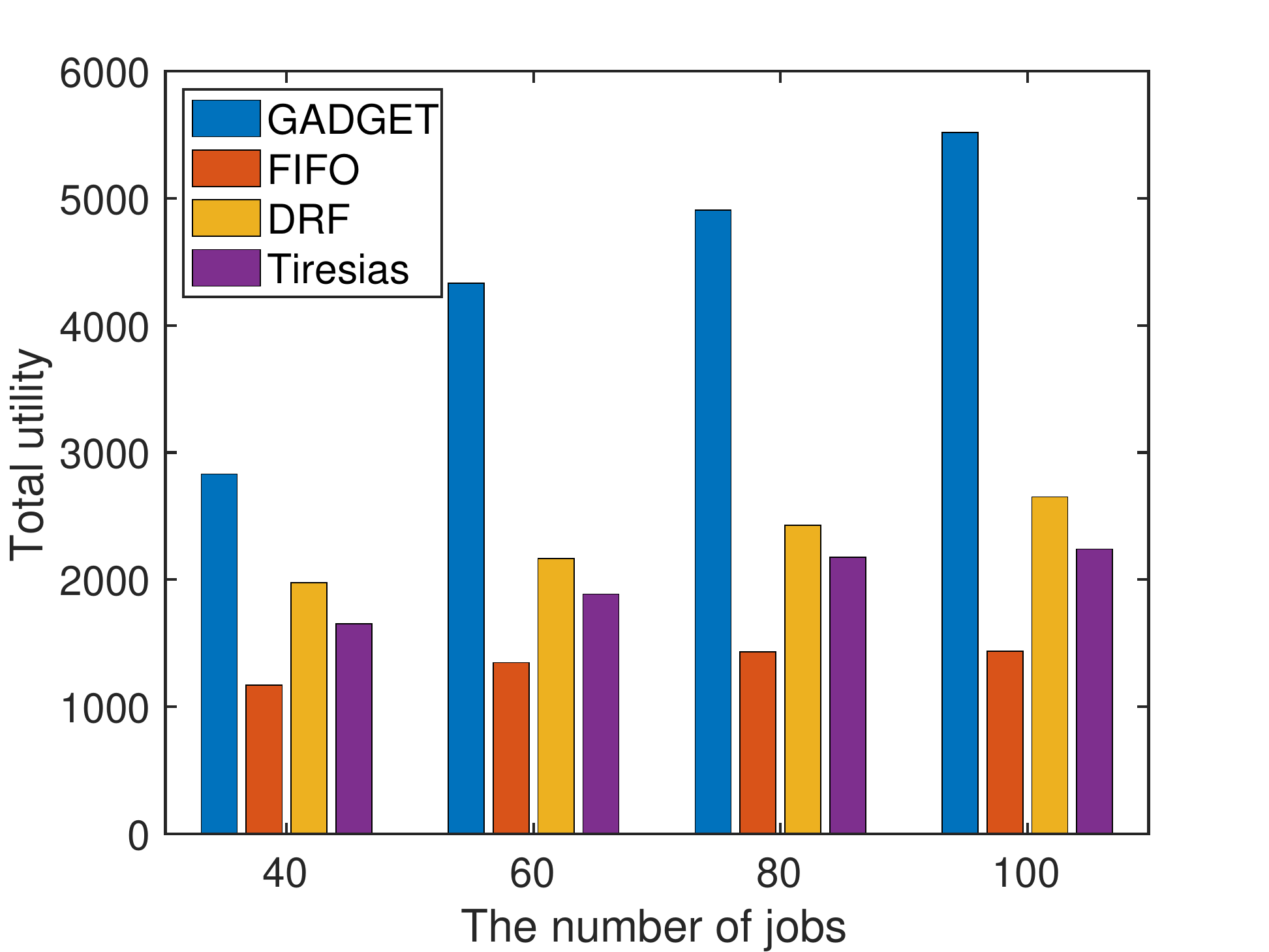}
        \caption{Total utility comparisons.} \label{fig:utility}
    \end{minipage}%
    \hspace{0.005\linewidth}
    \begin{minipage}[t]{0.24\linewidth}
        \centering
        \includegraphics[width=.96\textwidth]{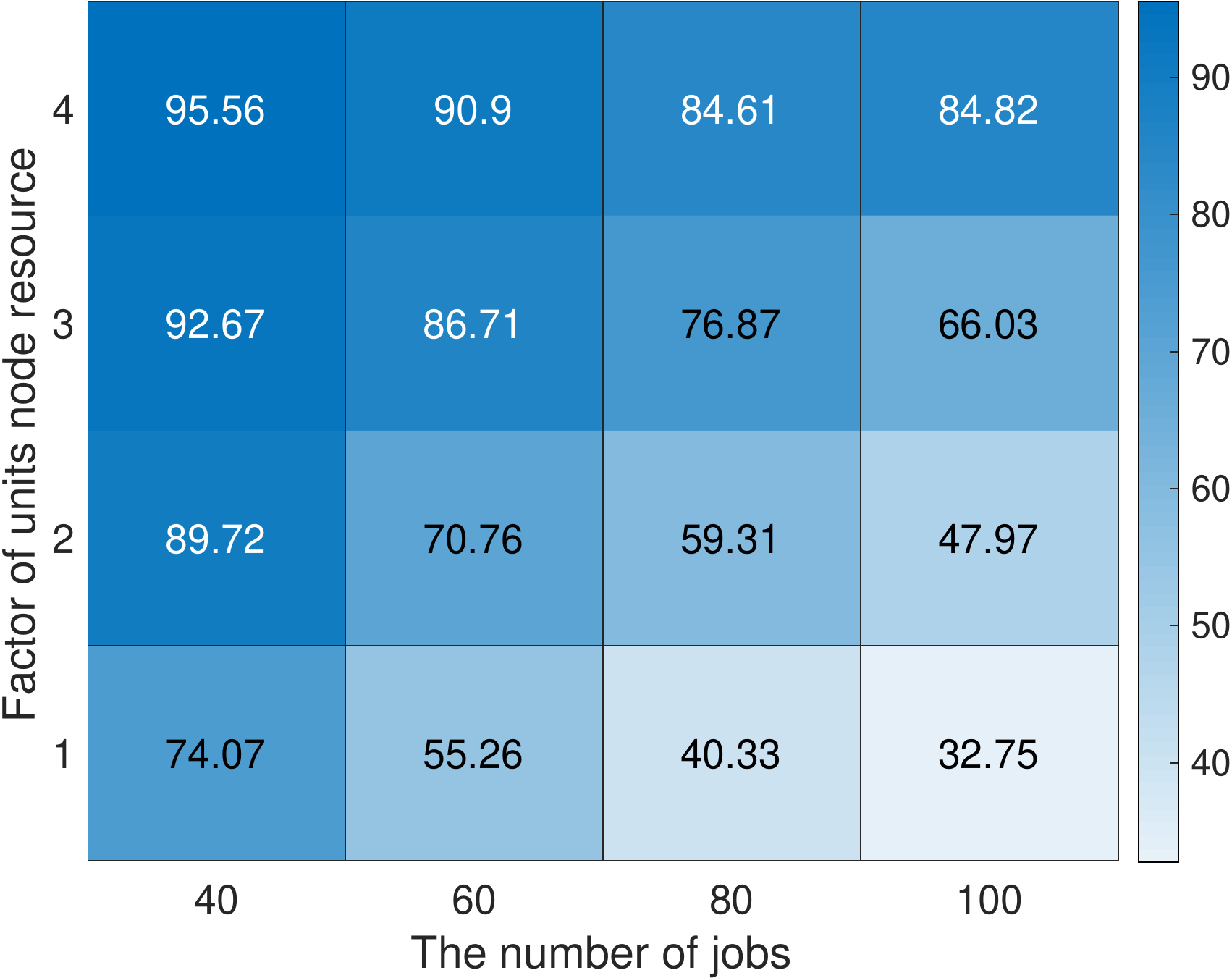}
        \caption{Embedded ratio overview w.r.t node resource.} \label{fig:NRF}
    \end{minipage}%
    \hspace{0.005\textwidth}
    \begin{minipage}[t]{0.24\linewidth}
        \includegraphics[width=1.1\textwidth]{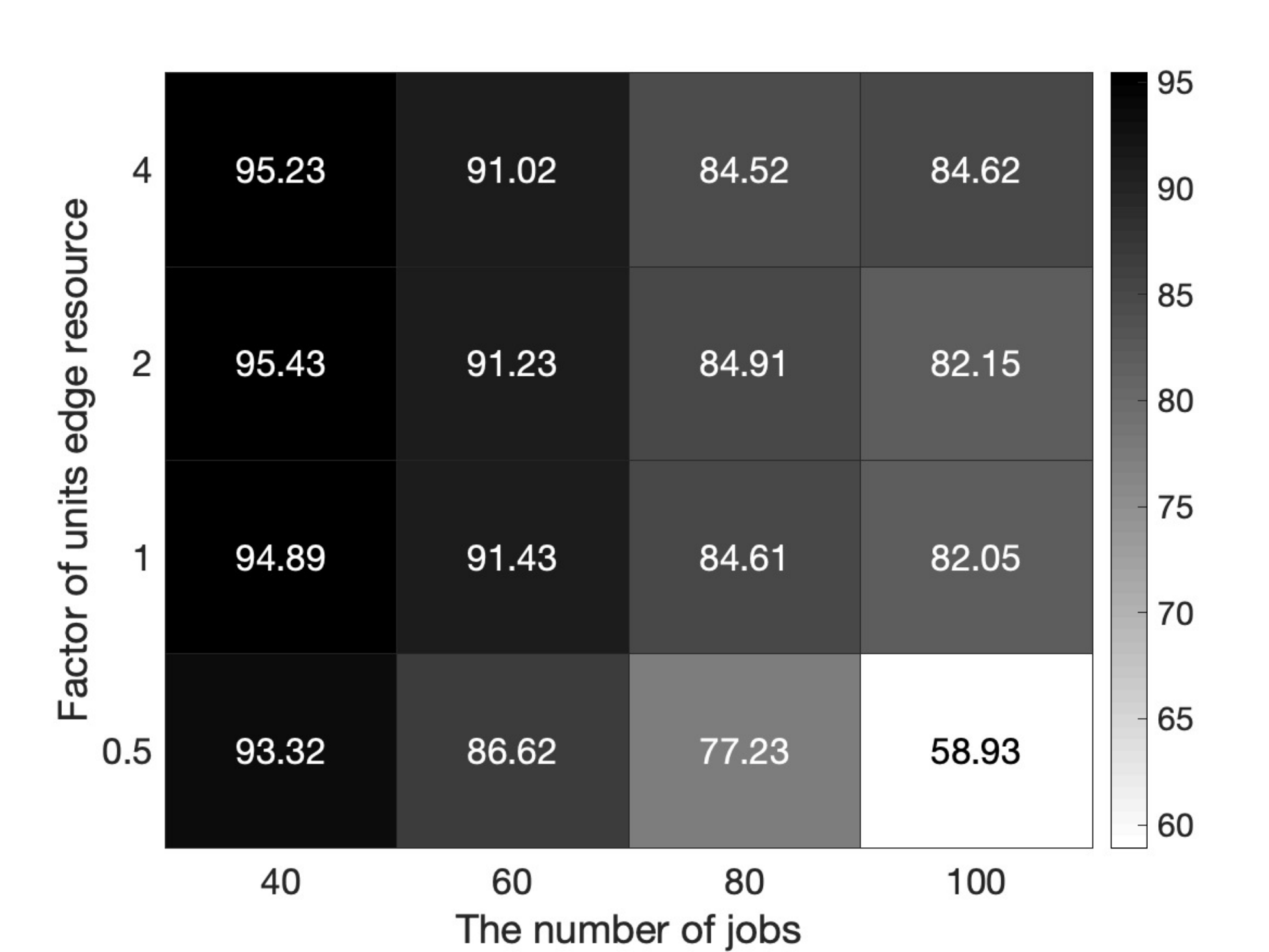}
        \caption{Embedded ratio overview w.r.t edge resource.}\label{fig:ERF}
    \end{minipage}%
    \hspace{0.005\linewidth} 
    \begin{minipage}[t]{0.24\linewidth}
        \includegraphics[width=1.1\textwidth]{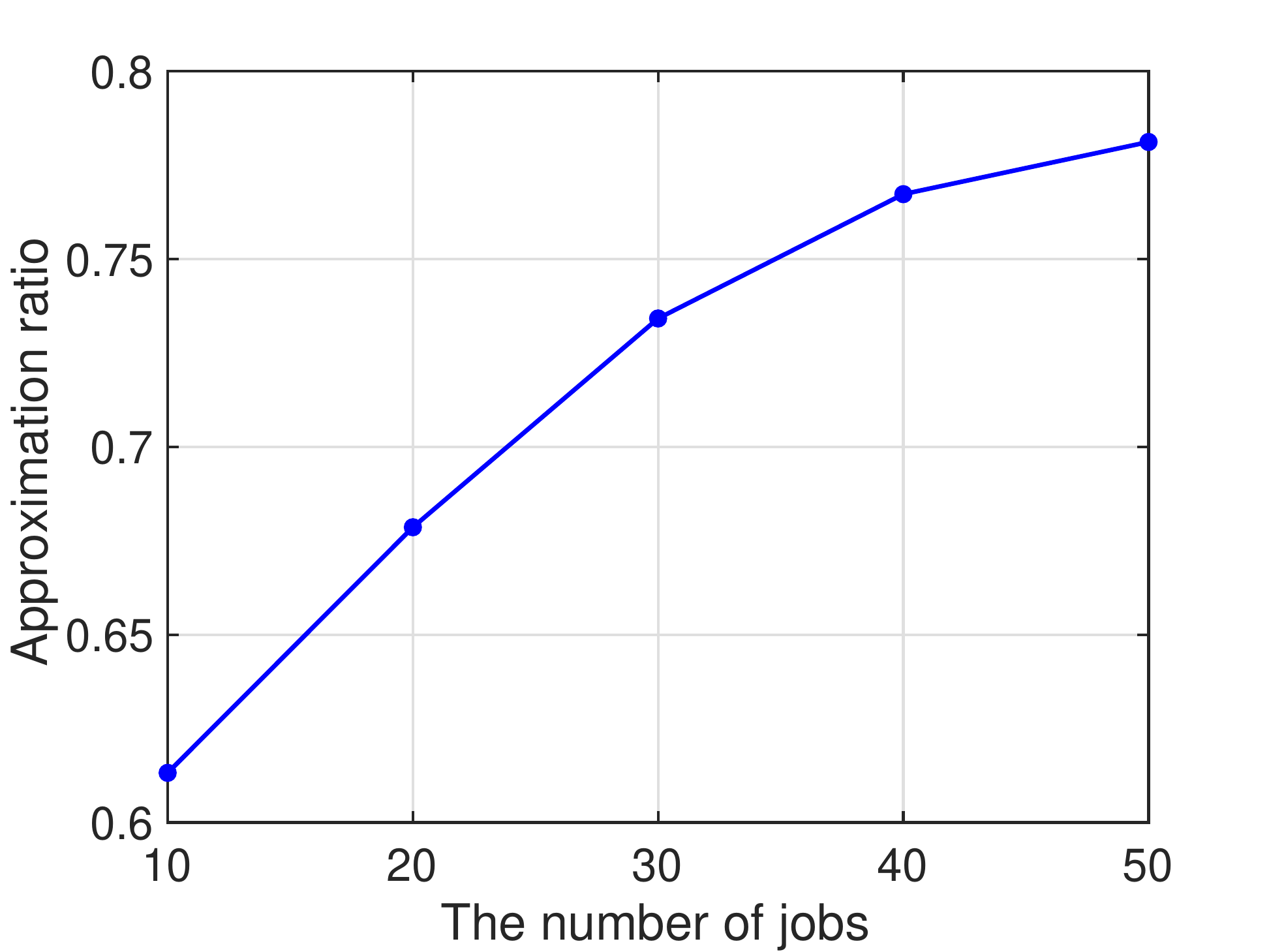}
        \caption{
		Approximation ratio.}\label{fig:ratio}
    \end{minipage}%
\vspace{-.2in}
\end{figure*}

\section{Empirical Studies}
\label{sec:numerical}

In this section, we conduct simulations to evaluate our GADGET algorithm.
We use ``excessive training avoidance'' as an application example (cf. Section~\ref{sec:general_framework}, Page~5), which aims to maximize the overall utility in a GPU computing cluster.
Here, Eq.~(\ref{ctr:iter}) is specialized to $\sum_{t}\sum_{s}y_{is}[t]\leq F_i,\forall i$, where $F_i$ is the maximum number of iterations specified by users upon their job submissions.
Eq.~(\ref{ctr:cap}) is specialized to $\sum_i y_{is}[t] \!\leq\! C_s, \forall s, t$, where $C_s$ is the GPU capacity of server $s$.

\smallskip
{\bf 1) Experiment Settings:} We use a ``fat-tree'' GPU computing cluster running for $T\!=\!200$, with $S\!=\!50$ servers.
The servers are randomly divided into ``racks'' and the number of racks is randomly chosen between $[2, 5]$.
Following similar settings as in~\cite{Mahajan20:Themis}, we configure each server with the number of GPUs randomly chosen from the discrete set $\{1, 2, 4, 8\}$.
We adopt the job arrival pattern from the Google Cluster data~\cite{Reiss12:googlecluster}.
The job parameters are integers generated uniformly at random from the following intervals: $N_i\!\!\in\!\![1, 5]$, $F_i\!\in\![1000,6000]$, $\zeta_i\!\in\![50,500]$, and $b_i\in[100 \text{ Mbps}, 5 \text{ Gbps}]$.
The bandwidths between racks and servers are chosen from $[200 \text{ Gbps}, 3200 \text{ Gbps}]$~\cite{Scott:bandwidth} and $[10 \text{ Gbps}, 100 \text{ Gbps}]$ uniformly at random, respectively.
We adopt the Sigmoid utility function~\cite{Huang15:CORA}: $\mu_i(y_{is}[t]) = \frac{\lambda_{i,1}}{1+e^{-\lambda_{i,2}(\zeta_i\sum_{s\in\mathcal{S}}y_{is}[t]+\zeta_iz_{i,t-1}-\lambda_{i,3})}}$, where $\lambda_{i,1}\in[1,100]$ represents the priority of job $i$, $\lambda_{i,2} \!\in\!(0,1)$ is to represent the sensitivity of the job to the number of iterations, and $\lambda_{i,3}\in[300, 3000]$ represents the expected number of iterations that should be trained.

\smallskip
{\bf 2) Baselines:} We compare GADGET with the following representative job scheduling policies for computing clusters:

\begin{list}{\labelitemi}{\leftmargin=1em \itemindent=-0.09em \itemsep=.2em}

\item {\em FIFO in Hadoop and Spark~\cite{Zaharia10:Spark}:} Jobs are scheduled in the order of their arrivals with a fixed number of workers. 

\item {\em Dominant Resource Fairness Scheduling (DRF) in Yarn~\cite{Vavilapalli13:Yarn} and Mesos~\cite{Hindman11:Mesos}:} Jobs are scheduled based on their dominant resource share.

\item {\em Least Attained Service (LAS) in Tiresias~\cite{Gu19:Tiresias}:} Resource scheduling among active users in the cluster is conducted in a round-robin fashion across jobs, according to the total number of accelerator hours consumed by each job.
\end{list}
Since the above scheduling policies do not consider the underlying topology constraints, we place workers based on the simple heuristic that greedily allocates workers to servers. 
where a cycle can be attained.
For FIFO and Tiresias, the number of workers is fixed to a number within $[1,10]$.

\smallskip
{\bf 3) Experiment Results:}
Fig.~\ref{fig:utility} illustrates the comparisons of our GADGET algorithm to the above baselines.
We can see that GADGET algorithm significantly outperforms the baseline algorithms, and the gains in total utility value over the baselines become more pronounced as the number of jobs increases.
This shows that the dynamic resource allocation in GADGET achieves higher resource utilization than those of the static resource allocations (i.e., the number of workers remains fixed throughout the training process) in the baselines.

Next, we examine the impacts of node GPU and edge bandwidth resource capacities in the substrate network on the performance of GADGET.
The results are shown in Figs.~\ref{fig:NRF} and \ref{fig:ERF}, where one unit of node and edge resource in the y-axes represents (GPU=100, bandwidth=200Gbps).
We evaluate the ratio between the numbers of embedded jobs and active jobs in each time slot, and the results in Figs.~\ref{fig:NRF} and \ref{fig:ERF} are the average ratio over three trials. 
We can see that increasing the node and edge capacities in the cluster has a positive impact, which allows more jobs to be embedded in each time slot.
Intuitively, as the GPU capacity of each node increases, jobs with high bandwidth demands have a higher probability to be allocated using intra-server communication, which typically has a much larger bandwidth capacity than that of inter-server communication.
Also, the larger the GPUs number, the more jobs can be trained simultaneously.
Similarly, the larger the edge capacity, the more jobs with inter-server communications can be scheduled.


Lastly, we investigate the performance of our proposed G-VNE technique, which is a major component of our GADGET algorithm.
We evaluate the embedding performance in terms of the ratio between the total utility obtained by our algorithm and the optimal total utility.
The optimal utility of Problem~(\ref{problem:mcf}) at each time slot is computed using the global optimization solver Gurobi based on the branch-and-bound approach (of exponential complexity) \cite{Kobayashi20:branch_and_bound}.
The results are shown in Fig.~\ref{fig:ratio}.
We can see that the actual performance ratio is better than our theoretical bound in Theorem~\ref{thm:ratio}, which achieves 60\%--80\% of the optimal utility obtained by the global optimization solver.

%% file: Sec7-Conclusion/Sec7-Conclusion.tex

\section{Conclusion}
\label{sec:conclusion}
In this paper, we studied online resource scheduling for the training of RAR-based DDL jobs in computing clusters.
We first developed an analytical optimization framework and then developed an online scheduling algorithm called GADGET
In GADGET, by showing the temporal submodularity of the online scheduling problem, we developed a greedy scheduling approach with competitive ratio guarantee.
Then, for the NP-Hard subproblem in each time-slot, we proposed a generalized virtual network embedding technique with approximation ratio guarantee.
Trace-driven simulations confirmed the superior performance of GADGET over existing schemes.~\footnote{The authors have provided public access to their code or data at \url{https://zenodo.org/record/5847644#.YervbxNKhTZ}.}